\documentclass[journal,twoside,web]{ieeecolor}

\usepackage{lcsys}
\usepackage{cite}
\usepackage{amsmath,amssymb,amsfonts, mathtools}
\usepackage{algorithmic}
\usepackage{graphicx}
\usepackage{textcomp}
\usepackage{bm}
\usepackage{ntheorem}

\newcommand{\E}{\mathbb{E}}
\newcommand{\hv}{\textsc{Vec}}
\renewcommand{\xi}{\psi}
\renewcommand{\omega}{\psi}
\renewcommand{\Omega}{\Psi}
\renewcommand{\rho}{\theta}
\renewcommand{\intercal}{\top}
\newcommand{\mathds}{\bm}
\newcommand{\reply}[1]{\textcolor{blue}{#1}}

\newtheorem{prop}{Proposition}
\newtheorem{lemma}{Lemma}
\newtheorem{cor}{Corollary}
\newtheorem{thm}{Theorem}
\newtheorem{definition}{Definition}

\usepackage{letltxmacro}
\LetLtxMacro\orgvdots\vdots
\LetLtxMacro\orgddots\ddots

\makeatletter
\DeclareRobustCommand\vdots{%
	\mathpalette\@vdots{}%
}
\newcommand*{\@vdots}[2]{%
	\sbox0{$#1\cdotp\cdotp\cdotp\m@th$}%
	\sbox2{$#1.\m@th$}%
	\vbox{%
		\dimen@=\wd0 %
		\advance\dimen@ -3\ht2 %
		\kern.5\dimen@
		\dimen@=\wd2 %
		\advance\dimen@ -\ht2 %
		\dimen2=\wd0 %
		\advance\dimen2 -\dimen@
		\vbox to \dimen2{%
			\offinterlineskip
			\copy2 \vfill\copy2 \vfill\copy2 %
		}%
	}%
}
\DeclareRobustCommand\ddots{%
	\mathinner{%
		\mathpalette\@ddots{}%
		\mkern\thinmuskip
	}%
}
\newcommand*{\@ddots}[2]{%
	\sbox0{$#1\cdotp\cdotp\cdotp\m@th$}%
	\sbox2{$#1.\m@th$}%
	\vbox{%
		\dimen@=\wd0 %
		\advance\dimen@ -3\ht2 %
		\kern.5\dimen@
		\dimen@=\wd2 %
		\advance\dimen@ -\ht2 %
		\dimen2=\wd0 %
		\advance\dimen2 -\dimen@
		\vbox to \dimen2{%
			\offinterlineskip
			\hbox{$#1\mathpunct{.}\m@th$}%
			\vfill
			\hbox{$#1\mathpunct{\kern\wd2}\mathpunct{.}\m@th$}%
			\vfill
			\hbox{$#1\mathpunct{\kern\wd2}\mathpunct{\kern\wd2}\mathpunct{.}\m@th$}%
		}%
	}%
}

\def\BibTeX{{\rm B\kern-.05em{\sc i\kern-.025em b}\kern-.08em
    T\kern-.1667em\lower.7ex\hbox{E}\kern-.125emX}}
\begin{document}
\title{Data-Driven Optimal Control of Affine Systems: \\ A Linear Programming Perspective}
\author{Andrea Martinelli, \IEEEmembership{Graduate Student Member, IEEE}, Matilde Gargiani, Marina Draskovic, \\ and John Lygeros, \IEEEmembership{Fellow, IEEE}
\thanks{Research supported by the European Research Council under the Horizon 2020 Advanced Grant No. 787845 (OCAL).}
\thanks{The authors are with the Automatic Control Laboratory, Swiss Federal Institute of Technology (ETH) Zurich, 8092 Zurich, Switzerland. (e-mail: andremar@ethz.ch, gmatilde@ethz.ch, mdraskovic@ethz.ch, lygeros@ethz.ch). 
}}

\maketitle
\thispagestyle{empty}

\begin{abstract}
In this letter, we discuss the problem of optimal control for affine systems in the context of data-driven linear programming. First, we introduce a unified framework for the fixed point characterization of the value function, \textit{Q}-function and relaxed Bellman operators. Then, in a model-free setting, we show how to synthesize and estimate Bellman inequalities from a small but sufficiently rich dataset. To guarantee exploration richness, we complete the extension of Willems' fundamental lemma to affine systems.
\end{abstract}


\section{Introduction}
\label{sec:introduction}

\IEEEPARstart{T}{he} linear programming (LP) approach to optimal control problems was initially developed by A.S. Manne in the 1960s \cite{ManneLP}, following the well-known studies conducted by R. Bellman in the 1950s \cite{BellmanDP}. The problem of deriving the fixed point of the Bellman operator can be cast as an LP by exploiting monotonicity and contractivity properties \cite{BertsekasAbstractDP}. An advantage of the LP formulation is that there exist efficient and fast algorithms to tackle such programs \cite{BoydConvexOptimization}. On the other hand, similarly to the classic dynamic programming approach introduced by Bellman, the LP approach suffers from poor scalability properties often referred to as \textit{curse of dimensionality} \cite{BertsekasNDP}. The sources of intractability for systems with continuous state and action spaces include the optimization variables in infinite dimensional spaces and infinite number of constraints. For this reason, the infinite dimensional LPs are usually approximated by tractable finite dimensional ones \cite{LasserreDTMCP,deFariasLPapproach,EsfahaniFromInftoFinitePrograms,PaulADP, LasserrePolynomials}.

In recent years, the LP approach has experienced an increasing interest, especially in combination with model-free control techniques \cite{SutterCDC2017,GoranADP,AlexandrosIFAC20,RelaxedBellmanOp}. In such a setting, one assumes the dynamical system to be unknown but observable via state-input trajectories, and builds one constraint (here called Bellman inequality) of the LP for each observed transition. In this way, one can bypass the classic system identification step and at the same time tackle a source of intractability by solving an LP with finite constraints. Empirical evidence suggests that the solution quality may dramatically depend on the number of sampled constraints \cite{RelaxedBellmanOp}. To avoid massive exploration, one can attempt to generate additional constraints offline from a small but sufficiently rich dataset. A preliminary analysis is conducted in \cite{AndreaSynthesisBellmanIneq} for linear systems in the value function formulation. Another fundamental problem is to estimate the expectation in the Bellman inequalities. A typical approach, \textit{e.g.} in \cite{SutterCDC2017} and \cite{RelaxedBellmanOp}, is to reinitialize the dynamics in the same state-input pairs and compute a Monte-Carlo estimate, even though this procedure could be unrealistic in stochastic settings.

Motivated by the poor scalability often affecting the LP approach and inspired by the recent literature revolving around Willems' fundamental lemma and data-driven control of affine systems, in the present work we discuss a unified framework to study data-driven control problems arising from this problem class. The authors in \cite{BoydGenQuadratics} show that an augmented state-space formulation allows one to tackle different affine stochastic control problems. The fundamental lemma \cite{WillemsPersistence} states that, for controllable \textit{linear} systems, persistency of excitation is a sufficient condition on the control signal to generate a trajectory that contains enough information to express any other trajectory of appropriate length as a linear combination of the input-output data. This result lies at the heart of many recent works on data-driven control of linear systems \cite{JeremyDeePC}, \cite{TesiDatadrivencontrol}, \cite{vanWaardeDataInformativity}. Extending the fundamental lemma to affine systems is not trivial, since the collected trajectories no longer form a linear subspace. To do so, we complement the initial result in \cite{BerberichFundLemmaAffineSystems} by a persistency of excitation argument, inspired by the state-space proof of the fundamental lemma described in \cite{HenkMultipleDatasets}.

\noindent Our main contributions can be summarised as follows: 

$\bullet \quad \mbox{In}$ Section \ref{sec:opt_control}, we introduce the stochastic optimal control framework for affine systems including the characterization of the fixed points corresponding to the value function, $Q$-function and relaxed Bellman operators;

$\bullet \quad \mbox{In}$ Section \ref{sec:Bellman_ineq}, we show how the Bellman inequalities can be reconstructed starting from a sufficiently rich dataset. Moreover, we provide estimators for the corresponding expectations that do not need reinitialization and show how to build LPs that preserve the optimal policy; 

$\bullet \quad \mbox{To}$ ensure the dataset is sufficiently rich, in Section \ref{sec:WillemFL} we extend Willems' fundamental lemma to affine systems by showing that controllability and persistency of excitation are still sufficient conditions to generate trajectories containing enough information. \smallskip

\noindent \textbf{Notation.} We denote with $\textsc{Tr}(\cdot)$, $\textsc{Spec}(\cdot)$, $\textsc{Rowsp}(\cdot)$, $\textsc{Vec}(\cdot)$ and $\bm{1}$ the trace, spectrum, row space and vectorization of a real matrix and a vector of ones of appropriate dimension. For a subset $\mathbb{Y}$ of a finite dimensional vector space, we denote with $\mathcal{S}(\mathbb{Y})$ the vector space of all real-valued measurable functions $g : \mathbb{Y} \rightarrow \mathbb{R}$ that have a finite weighted sup-norm \cite[\S 2]{BertsekasAbstractDP}, that is, $\|g(y)\|_{\infty,z} = \textstyle\sup_{y\in\mathbb{Y}}\frac{|g(y)|}{z(y)} < \infty$ with $z : \mathbb{Y} \rightarrow \mathbb{R}_{>0}$.

\section{Optimal Control of Affine Systems with Generalized Quadratic Stage-Cost}\label{sec:opt_control}

In this section, we introduce the stochastic optimal control problem for affine systems. We characterize the fixed point of the value function, $Q$-function and relaxed Bellman operators, discussing how their contraction properties allow one to express the fixed points via infinite-dimensional LPs. 

\subsection{Stochastic Optimal Control}\label{sec:stochoptcontrol}
Consider the following discrete-time affine dynamics,
\begin{equation}\label{system}
	x^{+} = f(x,u,\xi) =  Ax + Bu + c + \xi,
\end{equation}
with possibly infinite state and action spaces $x\in\mathbb{\mathbb{X}} \subseteq \mathbb{R}^{n}$ and $u \in \mathbb{U} \subseteq \mathbb{R}^{m}$. Here $\omega \in \mathbb{D} \subseteq \mathbb{R}^{n}$ denotes a random vector with possibly non-zero mean $\mu$ and covariance $\Sigma\succeq 0$. Moreover, $A\in\mathbb{R}^{n\times n}$ and $B \in 
\mathbb{R}^{n\times m}$ are such that $(\sqrt{\gamma}A, \sqrt{\gamma}B)$ is stabilizable, $c \in 
\mathbb{R}^{n}$ is a constant term and $\gamma \in (0,1)$ is a discount factor. We define a generalized quadratic stage-cost $\ell : \mathbb{X}\times\mathbb{U} \rightarrow \mathbb{R}_{\ge 0}$ as
\begin{equation}\label{stagecost}
	\ell(x,u) = \begin{bmatrix} x \\ u 	\end{bmatrix}^{\intercal} \underbrace{\begin{bmatrix} L_{xx} & L_{xu} \\ \star & L_{uu} 
	\end{bmatrix}}_{L} \begin{bmatrix} x \\ u \end{bmatrix} + 2\begin{bmatrix} x \\ u 	\end{bmatrix}^{\intercal} \underbrace{\begin{bmatrix} L_{x} \\ L_{u} \end{bmatrix}}_{L_{\ell}} + L_{c},
\end{equation}
where the symbol $\star$ is used to denote symmetry.
Consider the $\gamma$-discounted infinite-horizon cost associated to a deterministic stationary feedback policy $\pi : \mathbb{X} \rightarrow \mathbb{U}$,  
\begin{equation}\label{value function}
	v_{\pi}(x) = \mathbb{E} \left[ \sum_{k=0}^{\infty} \gamma^k \ell(x_k, \pi(x_k)) \; \Big| \; x_0=x \right].
\end{equation}
The objective of the optimal control problem is to find an optimal policy $\pi^*$ such that $v_{\pi^*} = \inf_{\pi}v_{\pi} = v^*$. Throughout the paper, to ensure that $v^* \in \mathcal{S}(\mathbb{X})$, $\pi^*$ is measurable and the infimum of $v_{\pi}$ is attained, we assume the stage-cost to be lower semi-continuous, nonnegative and inf-compact on $\mathbb{X} \times \mathbb{U}$, and that there exists a policy $\pi$ such that $v_\pi(x) < \infty$ for each $x \in \mathbb{X}$  \cite[Assumptions 4.2.1 and 4.2.2]{LasserreDTMCP}.

\subsection{Value Function Formulation}

The optimal value function $v^*$ is generally difficult to compute since, among other issues, it involves the minimization of an infinite sum of costs. However, it allows for the following recursive definition \cite{BellmanDP},
\begin{align}\label{Bellman operator for v functions}
	v^{*}(x) & = \inf_{u\in \mathbb{U}} \bigl\{ \ell(x,u) + \gamma \E\big[v^{*}(f(x,u,\xi))\big] \bigr\} \notag \\
	& = (\mathcal{T}v^{*})(x), 
\end{align}
where $\mathcal{T} : \mathcal{S}(\mathbb{X}) \rightarrow \mathcal{S}(\mathbb{X})$ is known as the \textit{Bellman operator}. $\mathcal{T}$ is a monotone, $\gamma$-contractive operator whose unique fixed point is $v^*$ \cite{DenardoContractionMaps, BertsekasAbstractDP}. 

When the dynamics is linear and the stage-cost quadratic, the resulting optimal control problem (LQR) enjoys well-known closed-form solutions \cite{DavisStochasticControl} based on the algebraic Riccati equations (ARE) $P^* = q^* - q^*_\ell q_c^{*-1} q^{*\intercal}_\ell$, where
\begin{equation*}\label{Q matrix}
	Q^*=\begin{bmatrix} q^* & q^*_\ell \\ \star & q^*_c
	\end{bmatrix} = \begin{bmatrix} L_{xx} + \gamma A^\intercal P^*A & L_{xu} + \gamma A^\intercal P^*B \\ \star & L_{uu} + \gamma B^\intercal P^*B
	\end{bmatrix}.
\end{equation*}
In case of affine systems, by suitably augmenting the system's coordinates, it is possible to show that the optimal policy also has an (affine) closed-form \cite{BoydGenQuadratics}. The next result characterizes $v^*$ and $\pi^*$ for affine systems and generalized quadratic stage-cost, by introducing notation and methods that will be reused in the extensions to $Q$-function and relaxed Bellman operator.
\begin{prop}\label{Proposition V-LQR}
	The fixed point of \eqref{Bellman operator for v functions} under dynamics \eqref{system}, stage-cost \eqref{stagecost} and $(\mathbb{X}, \mathbb{U}, \mathbb{D}) = (\mathbb{R}^n, \mathbb{R}^m, \mathbb{R}^n)$ is
	\begin{equation}\label{v star}
		v^*(x) = x^\intercal P^* x + 2x^\intercal P_{\ell}^* + P_{c}^* + \tfrac{\gamma \textsc{Tr}(P^*\Sigma)}{1-\gamma}.
	\end{equation}
	$\tilde{P}^* = \big[ \begin{smallmatrix} P^* & P_\ell^* \\ \star & P_c^*
	\end{smallmatrix} \big]$ is the unique positive definite solution to the following \textit{augmented} ARE
	\begin{equation}\label{augmented ARE}
		\tilde{P}^* = \tilde{q}^* - \tilde{q}^*_\ell \tilde{q}_c^{*-1} \tilde{q}_\ell^{*\intercal}, 
	\end{equation}
	\begin{equation}\label{augmented Q matrix}
		\begin{bmatrix} \tilde{q}^* & \tilde{q}^*_\ell \\ \star & \tilde{q}^*_c
		\end{bmatrix} = \begin{bmatrix} \tilde{L}_{xx} + \gamma \tilde{A}^\intercal \tilde{P}\tilde{A} & \tilde{L}_{xu} + \gamma \tilde{A}^\intercal \tilde{P}\tilde{B} \\ \star & L_{uu} + \gamma \tilde{B}^\intercal \tilde{P}\tilde{B}
		\end{bmatrix},
	\end{equation}
	$\tilde{A} = \big[ \begin{smallmatrix} A & c+\mu \\ 0 & 1 \end{smallmatrix} \big]$, $\tilde{B} = \big[ \begin{smallmatrix} B \\ 0 \end{smallmatrix} \big]$, $\tilde{L}_{xx} = \big[ \begin{smallmatrix} L_{xx} & L_x \\ \star & L_c
	\end{smallmatrix} \big]$, $\tilde{L}_{xu} = \big[ \begin{smallmatrix} L_{xu} \\ L_u^\intercal
	\end{smallmatrix} \big]$. \smallskip \\ 
	Finally, the associated optimal policy is
	\begin{equation}\label{optimal affine policy}
		\pi^*(x) = -q_{c}^{*-1} (q_{\ell}^{* \intercal} x + q),
	\end{equation}
	where $q = L_u + \gamma B^\intercal (P_\ell^* + P^*(c+\mu))$.

\end{prop}
\begin{proof} 
	Let us consider the constant update $y^+ = y$ initialized at $y_0 = 1$ and the augmented dynamics
	\begin{equation}\label{augmented system}
		\tilde{x}^+ = \tilde{A}\tilde{x}+\tilde{B}u+\tilde{\xi},
	\end{equation}
	where $\tilde{x} = \big[ \begin{smallmatrix} x \\ y \end{smallmatrix} \big]$ and $\tilde{\xi} = \big[ \begin{smallmatrix} \xi - \mu \\ 0 \end{smallmatrix} \big]$. Then,  $\mathbb{E}[\tilde{\xi}] = \big[ \begin{smallmatrix} 0 \\ 0 \end{smallmatrix} \big]$ and $\mathbb{E}[\tilde{\xi}\tilde{\xi}^\intercal] = \big[ \begin{smallmatrix} \Sigma & 0 \\ 0 & 0
	\end{smallmatrix} \big] = \tilde{\Sigma} \succeq 0$. Since $y=1 \;\, \forall t$, the stage-cost can be represented in augmented coordinates as 
	\begin{equation}\label{augmented stage-cost}
		\tilde{\ell}(\tilde{x},u) 
		= \begin{bmatrix} \tilde{x} \\ u \end{bmatrix}^\intercal \begin{bmatrix} \tilde{L}_{xx} & \tilde{L}_{xu} \\ \star &  L_{uu} \end{bmatrix} \begin{bmatrix} \tilde{x}\\u \end{bmatrix} = \ell(x,u).
	\end{equation}
	Stabilizability of $(\sqrt{\gamma}\tilde{A}, \sqrt{\gamma}\tilde{B})$ follows from stabilizability of the original pair. Indeed, $\textsc{Spec}(\tilde{A}) = \textsc{Spec}(A) \cup \{1\}$, and asymptotic stability of the uncontrollable mode $y^+=y$ is guaranteed by the discount factor.
	Hence, we can formulate a classic LQR problem in augmented coordinates whose unique solution is given by
	\begin{equation}\label{augmented value function}
		\tilde{v}^*(\tilde{x}) = \tilde{x}^\intercal \tilde{P}^* \tilde{x} + \tfrac{\gamma \textsc{Tr}(\tilde{P}^*\tilde{\Sigma})}{1-\gamma} = v^*(x),
	\end{equation} 
	where $\tilde{P}$ is the solution to the augmented ARE \eqref{augmented ARE}, and the associated optimal policy is 
	\begin{equation*}
		\begin{aligned}
			\tilde{\pi}^*(\tilde{x}) & = \arg \min_{u\in\mathbb{R}^{m}} \big\{ \tilde{\ell}(\tilde{x},u) + \gamma \E\big[\tilde{v}^*(\tilde{A}\tilde{x}+\tilde{B}u+\tilde{\xi})\big] \big\}  \\
			& = -\tilde{q}_{c}^{*-1}\tilde{q}_{\ell}^{*\intercal} \tilde{x}.
		\end{aligned}
	\end{equation*}  
	The claim then follows by writing the solution in the original coordinates. 
\end{proof}
Note that if the system is linear $(c=0)$, the noise zero-mean $(\mu=0)$ and the stage-cost pure quadratic $(L_\ell=0)$, then $q=0$ and we recover the linear policy $\pi^*(x) = -q_{c}^{*-1} q_{\ell}^{* \intercal} x$.

The fixed point of $\mathcal{T}$ can be computed via linear programming \cite{deFariasLPapproach}. By observing that the Bellman inequality $v \le \mathcal{T}v$ implies $v \le v^{*}$, one can characterize the fixed point of $\mathcal{T}$ by looking for the greatest $v \in \mathcal{S}(\mathbb{X})$ that satisfies $v \le \mathcal{T}v$,
\begin{equation}\label{NPvaluefunction}
	\begin{aligned}
		\sup_{v\in\mathcal{S}(\mathbb{X})} & \int_{\mathbb{X}} v(x)c(dx) \\
		\mbox{s.t.} \;\; & v(x) \le (\mathcal{T}v)(x) \quad \forall x \in \mathbb{X},
	\end{aligned}
\end{equation}
where $c(\cdot)$ is a positive measure with finite moments. In the LP literature, $c(\cdot)$ is typically selected to be a probability measure \cite{PaulADP,deFariasLPapproach}. For example, if the state-space is unbounded one can use a Gaussian distribution, if it is compact a uniform distribution. Moreover, the measure $c(\cdot)$ can be used to give different weight in the quality of the approximation of the value function in different parts of the state-space. Now, notice that $\mathcal{T}$ is a nonlinear operator. However, it is possible to reformulate \eqref{NPvaluefunction} as an equivalent linear program \cite{deFariasLPapproach} by dropping the infimum in $\mathcal{T}$ and substituting the nonlinear constraint set with the following linear one
\[ v(x) \le \underbrace{\ell(x,u) + \gamma \E\big[v(f(x,u,\xi))\big]}_{(\mathcal{T}_\ell v)(x,u)} \;\, \forall (x,u) \in \mathbb{X} \times \mathbb{U}.  \] 
The associated optimal policy can then be computed by
\begin{equation}\label{greedy policy}
	\pi^*(x) = \arg \min_{u\in\mathbb{U}} \big\{ \ell(x,u) + \gamma \E\big[v^*(f(x,u,\xi))\big] \big\}.
\end{equation}

\subsection{\textit{Q}-function Formulation}

Policy extraction \eqref{greedy policy} is in general not possible if the dynamics $f$ or the stage-cost $\ell$ are not known. This difficulty can be overcome by introducing the Bellman operator associated to $Q$-functions \cite{WatkinsQLearning}, $\mathcal{F} : \mathcal{S}(\mathbb{X} \times \mathbb{U}) \rightarrow \mathcal{S}(\mathbb{X} \times \mathbb{U})$,
\begin{align}\label{Bellman operator q functions}
	q^{*}(x,u) & = \ell(x,u) + \gamma \E\left[\inf_{w\in \mathbb{U}}q^*(f(x,u,\xi),w)\right] \notag \\
	& = (\mathcal{F}q^{*})(x,u).
\end{align}
The advantage of the $Q$-function reformulation is that policy extraction is model-free:
\begin{equation}
	\pi^*(x) = \arg \min_{u\in\mathbb{U}}q^*(x,u).
\end{equation}
In the following, we characterize the fixed point of $\mathcal{F}$ under affine dynamics and generalized quadratic stage-cost.
\begin{prop}
	The fixed point of \eqref{Bellman operator q functions} under dynamics \eqref{system}, stage-cost \eqref{stagecost} and $(\mathbb{X}, \mathbb{U}, \mathbb{D}) = (\mathbb{R}^n, \mathbb{R}^m, \mathbb{R}^n)$ is
	\begin{equation}\label{optimal q function}
		q^*(x,u) = \big[ \begin{smallmatrix} x \\ u \end{smallmatrix}\big]^{\intercal} Q^* \big[ \begin{smallmatrix} x \\ u \end{smallmatrix}\big] + 2\big[ \begin{smallmatrix} x \\ u \end{smallmatrix}\big]^{\intercal}Q^*_{\ell} + Q^*_c + \tfrac{\gamma \textsc{Tr}(P^*\Sigma)}{1-\gamma}.
	\end{equation} 
	By considering $Q^* = \begin{bsmallmatrix} q^* & q^*_\ell \\ \star & q^*_c
	\end{bsmallmatrix}$, $Q^*_\ell = \begin{bsmallmatrix} q^*_{x} \\ q^*_{u}
	\end{bsmallmatrix}$ and \eqref{augmented Q matrix}, it holds
	\begin{equation}
		\left[\begin{array}{cc|c} q^* & q^*_{x} & q^*_{\ell} \\
			\star & Q^*_c & q^{*\intercal}_{u} \\ \hline
			\star & \star & q^*_{c}
		\end{array}\right] = {\renewcommand{\arraystretch}{1.3}
			\left[\begin{array}{c|c} \tilde{q}^* & \tilde{q}^*_\ell \\ \hline \star & \tilde{q}^*_c
			\end{array}\right]}.
	\end{equation}
	The optimal policy is again given by \eqref{optimal affine policy}.
\end{prop}
\begin{proof}
	Similarly to Proposition \ref{Proposition V-LQR}, we can consider the augmented dynamics \eqref{augmented system} and augmented stage-cost \eqref{augmented stage-cost} and verify that \eqref{optimal q function} satisfies the fixed point equation \eqref{Bellman operator q functions}.
\end{proof}
Since the operator $\mathcal{F}$ shares the same monotonocity and contractivity properties of $\mathcal{T}$ \cite{BertsekasNDP}, we can write again a (nonlinear) exact program for the $Q$-function
\begin{equation}\label{NonlinearProgramQfunctions}
	\begin{aligned}
		\sup_{q\in\mathcal{S}(\mathbb{X} \times \mathbb{U})} & \int_{\mathbb{X} \times \mathbb{U}} q(x,u)c(dx,du) \\
		\mbox{s.t.} \;\;\;\; & q(x,u) \le (\mathcal{F}q)(x,u) \quad \forall (x,u) \in \mathbb{X} \times \mathbb{U},
	\end{aligned}
\end{equation}
where $c(\cdot, \cdot)$ takes the same role as in \eqref{NPvaluefunction}. Unlike \eqref{NPvaluefunction}, it is not straightforward to replace the nonlinear constraints in \eqref{NonlinearProgramQfunctions} with linear ones due to the nesting of the $\E$ and $\inf$ operators in \eqref{Bellman operator q functions}. A linear reformulation of \eqref{NonlinearProgramQfunctions} can be obtained, as shown in \cite{CogillDecentralizedADP} and \cite{PaulADP}, by introducing additional decision variables,
\begin{equation}\label{LPoldQfunctions}
	\begin{aligned}
		\sup_{\begin{smallmatrix}
				v\in \mathcal{S}(\mathbb{X}), \\ q\in\mathcal{S}(\mathbb{X} \times \mathbb{U})
		\end{smallmatrix}} & \int_{\mathbb{X} \times \mathbb{U}} q(x,u)c(dx,du) \\
		\mbox{s.t.} \;\;\;\; & q(x,u) \le (\mathcal{T}_\ell v)(x,u) \quad \forall (x,u) \in \mathbb{X} \times \mathbb{U} \\
		& v(x) \le q(x,u) \quad \forall (x,u) \in \mathbb{X} \times \mathbb{U}.
	\end{aligned}
\end{equation}
In the next section, we show how to reduce the number of decision variables by introducing a modified operator. 

\subsection{Relaxed Bellman Operator Formulation}

The authors in \cite{RelaxedBellmanOp} introduce the relaxed Bellman operator $\mathcal{\hat{F}} : \mathcal{S}(\mathbb{X} \times \mathbb{U}) \rightarrow \mathcal{S}(\mathbb{X} \times \mathbb{U})$,
\begin{equation}\label{relaxed Bellman op}
	(\hat{\mathcal{F}}\hat{q})(x,u) = \ell(x,u) + \gamma \inf_{w\in\mathbb{U}} \mathbb{E} [\hat{q}(f(x,u,\xi), w)],
\end{equation}
which retains the same structure of \eqref{Bellman operator q functions}, but with the infimum and expectation operators exchanged. The next result extends \cite[Theorem 3]{RelaxedBellmanOp} to affine dynamics and generalized quadratic stage-cost, showing that the fixed point of $\hat{\mathcal{F}}$ is again an upper estimator of the fixed point of $\mathcal{F}$ that preserves the minimizer with respect to $u$, \textit{i.e.} the optimal policy.
\begin{prop}
	The fixed point of \eqref{relaxed Bellman op} under dynamics \eqref{system}, stage-cost \eqref{stagecost} and $(\mathbb{X}, \mathbb{U}, \mathbb{D}) = (\mathbb{R}^n, \mathbb{R}^m, \mathbb{R}^n)$ is
	\begin{equation}\label{optimal q hat}
		\hat{q}(x,u) = q^{*}(x,u) + \tfrac{\gamma \textsc{Tr}(q^*_\ell q^{* -1}_c q^{*\intercal}_\ell \Sigma)}{1-\gamma},
	\end{equation} 
	and the optimal policy is again given by \eqref{optimal affine policy}.
\end{prop}
\begin{proof}
	Following the main steps of the proof of Theorem 3 in \cite{RelaxedBellmanOp}, one can verify that \eqref{optimal q hat} satisfies the fixed point equation $\hat{q} = \hat{\mathcal{F}}\hat{q}$ and, by uniqueness of the fixed point of $\hat{\mathcal{F}}$, conclude the proof.
\end{proof}
The relaxed operator $\hat{\mathcal{F}}$ shows significant computational advantage with respect to $\mathcal{F}$ when used in the LP formulation \cite{RelaxedBellmanOp}. In fact, since $\hat{\mathcal{F}}$ is also a monotone contraction mapping, its unique fixed point can be computed via a relaxation of \eqref{LPoldQfunctions} with reduced decision variables and constraints,
\begin{equation}\label{alternativeLPQfunctions}
	\begin{aligned}
		\sup_{q\in\mathcal{S}(\mathbb{X}\times\mathbb{U})} & \int_{\mathbb{X}\times\mathbb{U}} q(x,u)c(dx,du) \\
		\mbox{s.t.} \;\;\;\; & q(x,u) \le (\hat{\mathcal{F}}_{\ell}q)(x,u,w) \quad \forall (x,u,w) \in \mathbb{X}\times\mathbb{U}^2,
	\end{aligned}
\end{equation}
where $(\hat{\mathcal{F}}_{\ell} q)(x,u,w) = \ell(x,u) + \gamma \mathbb{E}\left[q(f(x,u,\xi),w)\right]$. The relaxed LP \eqref{alternativeLPQfunctions} and the fixed point characterization \eqref{optimal q hat} constitute the starting point for the next discussion on how to synthesize Bellman inequalities from data.

\section{Synthesis of Bellman Inequalities from Data}\label{sec:Bellman_ineq}

In the data-driven context, two fundamental problems in the LP formulation are the synthesis of Bellman inequalities from data (to avoid massive exploration of the state-space) and the estimation of expected values. In \cite{AndreaSynthesisBellmanIneq}, a preliminary study on linear systems in the value function formulation is conducted. Here, we generalize the analysis to affine systems in the relaxed Bellman operator formulation. Moreover, we provide novel estimators for the expectation in the constraints that do not require reinitialization and discuss how to employ them to build LPs that preserve the optimal policy.

Consider the family of generalized quadratic functions
\begin{equation*}\label{gen quad parametrization}
	\mathcal{Q} =\big\{ q(x,u)=\big[ \begin{smallmatrix} x \\ u \end{smallmatrix}\big]^{\intercal} Q \big[ \begin{smallmatrix} x \\ u \end{smallmatrix}\big] + 2\big[ \begin{smallmatrix} x \\ u \end{smallmatrix}\big]^{\intercal}Q_{\ell} + Q_c, \, Q=Q^\intercal \big\},
\end{equation*}
and note that $\mathcal{Q} \in \mathcal{S}(\mathbb{X} \times \mathbb{U})$. Then, define $m_{c}\in\mathbb{R}_{\ge0}$, $\mu_c\in\mathbb{R}^{n+m}$ and $\Sigma_c\succeq0$ as the zeroth, first and second raw moment (\textit{i.e.} centered about zero) of the measure $c(\cdot,\cdot)$.
\begin{prop}
	When $q(x,u) \in \mathcal{Q}$, the LP \eqref{alternativeLPQfunctions} takes the form 
\begin{equation}\label{linear formulation constraint set}
			\begin{aligned}
				\max_{\varphi} & \quad \begin{bsmallmatrix} (\hv\,\Sigma_c)^\intercal & 2\mu_c^\intercal & m_c
				\end{bsmallmatrix} \, \varphi \\
				\emph{s.t.} & \;\;\; \E\left[\rho(x,u,\omega,w)\right]^\intercal \varphi \le \ell(x,u),
			\end{aligned}
	\end{equation}
	for all $(x,u,w) \in \mathbb{X}\times\mathbb{U}^2$, where
	\begin{equation*}
		\theta = \begin{bsmallmatrix}
				\hv\big( \big[ \begin{smallmatrix} x \\ u \end{smallmatrix}\big] \big[ \begin{smallmatrix} x \\ u \end{smallmatrix}\big]^\intercal -\gamma \big[ \begin{smallmatrix} x^+ \\ w \end{smallmatrix}\big] \big[ \begin{smallmatrix} x^+ \\ w \end{smallmatrix}\big]^\intercal \big) \vspace{0.1cm}\\ 
				2 \big[ \begin{smallmatrix} x \\ u \end{smallmatrix}\big] - 2 \gamma \big[ \begin{smallmatrix} x^+ \\ w \end{smallmatrix}\big] \vspace{0.05cm}\\
				1-\gamma
			\end{bsmallmatrix}, \quad \varphi = \begin{bsmallmatrix} \hv\,Q \\ Q_{\ell} \\ Q_{c}
			\end{bsmallmatrix}.
	\end{equation*}
	%
\end{prop}
\begin{proof}
	Any quadratic form can be decomposed as
	\begin{equation}\label{quadratic form decomposition}
		\big[ \begin{smallmatrix} x \\ u \end{smallmatrix}\big]^{\intercal} Q \big[ \begin{smallmatrix} x \\ u \end{smallmatrix}\big] = \hv\big( \big[ \begin{smallmatrix} x \\ u \end{smallmatrix}\big] \big[ \begin{smallmatrix} x \\ u \end{smallmatrix}\big]^\intercal\big)^\intercal\, \hv\, Q .
	\end{equation}
	We obtain \eqref{linear formulation constraint set} by rearranging the constraints in \eqref{alternativeLPQfunctions} as $\E [ q(x,u) - \gamma q(x^+,w) ] \le \ell(x,u)$, imposing $q(x,u) \in \mathcal{Q}$ and \eqref{quadratic form decomposition} and performing the integration in the objective. 
\end{proof}
\begin{definition}[Dataset]
	When $X\in\mathbb{R}^{n\times d}$, $U\in \mathbb{R}^{m\times d}$ and $\Omega\in\mathbb{R}^{n\times d}$ are a collection of states, inputs and noise realizations and $X^+=AX+BU+c\mathds{1}^\intercal + \Psi$, we say that $(X,U,X^+)$ is a dataset of length $d$. A dataset corresponds to a single trajectory when $X_{i+1} = X^+_{i}$ for all $i=1,\ldots,d-1$, where $X_i$ denotes the $i$-th column of $X$. To specify a length different from $d$ we use the notation $X_{1:h} = \begin{bsmallmatrix} X_1 & \cdots & X_h \end{bsmallmatrix}$.
\end{definition}
In order to estimate the expected values in the Bellman inequalities \eqref{linear formulation constraint set}, as discussed \textit{e.g.} in \cite{RelaxedBellmanOp} and \cite{SutterCDC2017}, one could reinitialize the dynamics at a fixed state-input pair $(x,u)$ a number of times $d$, observe the corresponding transition $f(x,u,\Psi_i)$ and compute the unbiased estimator $\hat{\rho} = \frac{1}{d}\sum_{i=1}^{d}\rho(x,u,\Psi_i,w)$, such that $\E[\hat{\rho}] = \E[\rho(x,u,\omega,w)]$ and $\textsc{Var}( \hat{\rho} ) = \tfrac{1}{d}\textsc{Var}(\rho(x,u,\omega,w))$.
On the other hand, such an estimation can only be performed if one can reinitialize the dynamics at the same state $x$ and play the same input $u$ multiple times. Since this might be unrealistic in a stochastic framework, we show the effect of removing the reinitialization assumption by averaging the observations over the data directly instead of over the vectors $\rho(x,u,\Omega_i,w)$.
\begin{lemma}\label{lemma dataset}
	Consider a dataset $(X,U,X^+)$ of length $d$ and a matrix $W\in\mathbb{R}^{m \times d}$ such that 
	\begin{equation}\label{rank condition}
		\textsc{Rank} \left[\begin{smallmatrix}
			X \\ U \\ \mathds{1}^\intercal \\ W
		\end{smallmatrix} \right] = n +2m +1.
	\end{equation} 
	Then, $\forall (x,u,w) \in \mathbb{X}\times\mathbb{U}^2$, there exists $\alpha \in \mathbb{R}^{d}$ satisfying
	\begin{equation}\label{system of linear equations}
			\left[\begin{smallmatrix}
				X \\ U \\ \mathds{1}^\intercal \\ W 
			\end{smallmatrix}\right] \alpha = \left[\begin{smallmatrix}
				x \\ u \\ 1 \\ w
			\end{smallmatrix}\right],
	\end{equation} 
	and an estimator $\bar{\rho} = \rho\left( X\alpha, U\alpha, \Omega \alpha, W\alpha \right)$ such that
	\begin{enumerate}
		\item[(i)]\label{bellman ineq reconstr}
		$\bar{\rho} = \rho(x,u,\Omega \alpha,w)$,
		
		\item[(ii)] $\bar{\rho}$ has mean $\E \left[ \bar{\rho} \right] = \E \left[\rho(x,u,\bar{\omega},w)\right]$ and covariance $\textsc{Var}\left( \bar{\rho} \right) = \|\alpha\|^2_2 \textsc{Var}\left(\rho(x,u,\omega,w)\right)$, where $\bar{\omega}$ is a random vector with mean $\E[\bar{\omega}] = \mu$ and covariance $\bar{\Sigma} = \|\alpha\|^2_2\Sigma$.  
	\end{enumerate}
\end{lemma}
\begin{proof}
	(i) First, note that the rank-condition \eqref{rank condition} implies that \eqref{system of linear equations} always has a solution. Then, we have $X^+ \alpha = (AX+BU+c\mathds{1}^\intercal+\Omega)\alpha = Ax+Bu+c+\Omega\alpha$. Finally, the result holds by substituting \eqref{system of linear equations} and $X^+\alpha$ into the definition of $\rho(x,u,\omega,w)$. 
	
	(ii) Let us define $\psi_i$, $i=1,\ldots,d$ as independent random vectors with mean $\mu$ and covariance $\Sigma$. Then, $\bar{\psi} = \begin{bsmallmatrix} \psi_1 & \cdots & \psi_d \end{bsmallmatrix} \alpha$ is also a random vector. In particular, its mean is $\E[\bar{\psi}] = \mathds{1}^\intercal \alpha \mu = \mu$ and, since $\textsc{Cov}(\psi_i,\psi_j) = 0$ for all $i \ne j$ due to independence, its covariance is $\textsc{Var}(\bar{\psi})= \| \alpha \|^2_2 \Sigma = \bar{\Sigma}$. The claim then follows by considering (i).
\end{proof}
Note that if the underlying dynamics is deterministic (\textit{i.e.} $\Omega=0$), the estimator reduces to $\bar{\rho} = \rho(x,u,0,w)$. Then, if \eqref{rank condition} holds, due to the lack of expectation in \eqref{linear formulation constraint set} one can potentially reconstruct all infinite constraints by taking linear combination of the data, similarly to the discussion in \cite{AndreaSynthesisBellmanIneq} for linear systems in the value function formulation. 
%
%
%

In general, $\bar{\rho}$ is an unbiased estimator of $\E \left[\rho(x,u,\bar{\omega},w)\right]$, instead of $\E\left[\rho( x, u, \omega, w)\right]$. The former is the coefficient associated to the Bellman inequalities of the dynamical system $x^+=Ax+Bu+c+\bar{\omega}$. By inspecting \eqref{optimal affine policy}, we note that the optimal policy $\pi^*(x)$ depends on $\E[\omega]$ but not on $\Sigma$; the latter appears only in the constant terms of $v^*(x)$, $q^*(x,u)$ and $\hat{q}(x,u)$ (see \eqref{v star}, \eqref{optimal q function} and \eqref{optimal q hat}). 
Hence, for $(\mathbb{X}, \mathbb{U}, \mathbb{D}) = (\mathbb{R}^n, \mathbb{R}^m, \mathbb{R}^n)$, the solution to \eqref{linear formulation constraint set} associated to $\bar{\psi}$ is
\begin{equation*}
	\bar{q}(x,u) = \hat{q}(x,u) + \tfrac{\gamma \textsc{Tr}((\| \alpha \|^2_2 -1)q^*\Sigma)}{1-\gamma},
\end{equation*} and the associated optimal policy remains \eqref{optimal affine policy}. 

In summary, the estimator $\bar{\rho}$ can be computed from system's trajectories, it does not require reinitialization and can be used to construct LPs with biased constraints that preserve the optimal policy. Note that, to implement the approximation described above, one has to rely on estimators with the same covariance, \textit{i.e.} same $\| \alpha \|^2_2$. The study of statistical bounds due to constraint approximation is deferred to future research, while the interested reader is referred to \cite{deFariasConstraintSampling} and \cite{EsfahaniFromInftoFinitePrograms} for a discussion on error bounds due to constraint sampling and randomized optimization in the LP framework. Finally, please refer to \cite{RelaxedBellmanOp} for a description on how to implement the LPs described above in a model-free fashion and on the observed control performance.

\section{Willems' Fundamental Lemma for Affine Systems}\label{sec:WillemFL}

The previous discussion on how to synthesize Bellman inequalities from data is based on the rank assumption \eqref{rank condition}. While the importance of \eqref{rank condition} has been recognised in the affine systems literature \cite{BerberichFundLemmaAffineSystems}, conditions on the inputs applied to the system that ensure \eqref{rank condition} have proved elusive.  

In \cite{WillemsPersistence}, J.C. Willems and co-authors discuss how the information contained in a sufficiently rich and long trajectory of a \textit{linear} system is enough to describe any other trajectory of appropriate length that can be generated by the system. This result is known as Willems' fundamental lemma. The sufficient conditions such that the generated trajectory is rich enough are that (i) the system is controllable, and (ii) the input sequence is \textit{persistently exciting} of sufficient order. A sequence $S_1,\ldots,S_d \in \mathbb{R}^m$, $S = \begin{bmatrix} S_1 & \cdots & S_d \end{bmatrix}$, is persistently exciting of order $K\in\mathbb{N}_{>0}$ if the associated Hankel matrix of depth $K$,
\begin{equation*}
	\mathcal{H}_K(S) = \begin{bsmallmatrix} S_1 & S_2 & \cdots & S_{d-K+1} \\
		S_2 & S_3 & \cdots & S_{d-K+2} \\
		\vdots & \vdots &  & \vdots \\
		S_{K} & S_{K+1} & \cdots & S_d
	\end{bsmallmatrix} \in \mathbb{R}^{mK \times (d-K+1)},
\end{equation*}  
is full row-rank, \textit{i.e.} $\textsc{Rank}(\mathcal{H}_K(S)) = mK$. A necessary condition for this is $d \ge (m+1)K -1$.

\begin{prop}\label{prop:pe}
	If $S$ is persistently exciting of order $K$, then for all $K'<K$ it holds that (i) $S$ is persistently exciting of order $K'$ and (ii)  $\mathds{1}^\intercal \notin \textsc{Rowsp}\,\mathcal{H}_{K'}(S)$.
\end{prop}
\begin{proof}
	(i) The rows of $\mathcal{H}_K(S)$ are linearly independent, therefore $\mathcal{H}_{K'}(S)_{1:d-K+1}$ is full row-rank and so is $\mathcal{H}_{K'}(S)$.
	
	(ii) Assume $\mathds{1}^\intercal \in \textsc{Rowsp}\,\mathcal{H}_{K'}(S)$ by contradiction. Then, there exists $\beta \in \mathbb{R}^{mK'}$ so that $\beta^\intercal \mathcal{H}_{K'}(S) = \mathds{1}^\intercal$. Given the Hankel structure, it holds $\begin{bsmallmatrix}
		\beta \\ 0 \end{bsmallmatrix}^\intercal \mathcal{H}_{K}(S) = \begin{bsmallmatrix}
		0 \\ \beta \end{bsmallmatrix}^\intercal \mathcal{H}_{K}(S) = \mathds{1}^\intercal$ and hence $\mathcal{H}_{K}(S)$ is not full row-rank.
\end{proof}

The authors in \cite{BerberichFundLemmaAffineSystems} discuss how, in the affine case, the additional constraint $\mathds{1}^\intercal \alpha = 1$ is necessary to ensure that system trajectories can be expressed as a linear combination of the collected data. In light of Lemma \ref{lemma dataset}, it is evident that this constraint allows one to average out the affine term when taking linear combinations of the trajectories. To extend the fundamental lemma to affine systems, it remains to derive sufficient conditions such that \eqref{rank condition} is guaranteed. In the following, we consider an affine output transformation $y = Cx + Du + r$, with $C\in \mathbb{R}^{z\times n}$, $D\in\mathbb{R}^{z \times m}$ and $r\in\mathbb{R}^{z}$, and deterministic dynamics (\textit{i.e.} $\Omega=0$).

\begin{thm}[Fundamental lemma for affine systems]\label{fundamental lemma for affine sys} 
	Consider a dataset $(X,U,X^+,Y)$ corresponding to a single trajectory of length $d$ with $X^+=AX+BU+c\mathds{1}^\intercal$ and $Y = CX+DU+r\mathds{1}^\intercal$. If $U$ is a persistently exciting input of order $n+L+1$ and $(A,B)$ is a controllable pair, then
	\begin{enumerate}
		\item[(i)] $\textsc{Rank}\begin{bsmallmatrix} \mathcal{H}_{1}(X_{1:d-L+1}) \\ \mathcal{H}_{L}(U) \\ \mathds{1}^\intercal
		\end{bsmallmatrix} = n+mL+1$, \vspace{0.15cm}
		\item[(ii)] $(\tilde{X},\tilde{U},\tilde{X}^+,\tilde{Y})$ is a dataset of length $L$ if and only if there exists $g \in \mathbb{R}^{d-L+1}$ such that
		\begin{equation*}
			\begin{bsmallmatrix}
				\textsc{Vec} \, \tilde{U} \\ \textsc{Vec} \, \tilde{Y} \\ 1
			\end{bsmallmatrix} = \begin{bsmallmatrix}
				\mathcal{H}_L(U) \\ \mathcal{H}_L(Y) \\ \mathds{1}^\intercal
			\end{bsmallmatrix} g.
		\end{equation*}
	\end{enumerate}
	
\end{thm}
\begin{proof}
	The fact that (i) $\Rightarrow$ (ii) was proven in \cite[Theorem 1]{BerberichFundLemmaAffineSystems}. In the following, inspired by the state-space proof of the fundamental lemma for linear systems in \cite{HenkMultipleDatasets}, we prove that controllability and persistency of excitation of higher order imply (i). Since $U$ is persistently exciting of order $n+L+1$, by Proposition \ref{prop:pe} it is also persistently exciting of order $n+L$ and its associated Hankel matrix $\mathcal{H}_{n+L}(U) \in \mathbb{R}^{m(n+L) \times (d-n-L+1)}$ is full row-rank; in particular, $d \ge (m+1)(n+L)-1$. For compactness, let us denote the number of columns of $\mathcal{H}_{n+L}(U)$ as $h = d-n-L+1$. To establish this claim, we consider the row vectors $\nu \in \mathbb{R}^{1 \times n}$, $\eta \in \mathbb{R}^{1 \times mL}$ and $\epsilon \in \mathbb{R}$ such that $\left[\, \nu \;\; \eta \;\; \epsilon \, \right]$ is a vector in the left kernel of $\begin{bsmallmatrix}
		\mathcal{H}_{1}(X_{1:d-L+1})^\intercal & \mathcal{H}_L(U)^\intercal & \mathds{1} \end{bsmallmatrix}^\intercal$ and show that $\left[\, \nu \;\; \eta \;\; \epsilon \, \right] = 0$. We can write the first $h$ scalar equations in $\left[\, \nu \;\; \eta \;\; \epsilon \, \right] \begin{bsmallmatrix}
		\mathcal{H}_{1}(X_{1:d-L+1})^\intercal & \mathcal{H}_L(U)^\intercal & \mathds{1} \end{bsmallmatrix}^\intercal=0$ as $\nu X_i + \eta \textsc{Vec}\,U_{i:i+L} + \epsilon = 0$ for $i=1,\ldots,h$. For each scalar equation $i$, we can derive $n$ additional equations by repeatedly applying the laws of the system $X_{i+1} = AX_{i} + BU_{i} + c$, obtaining $(n+1)h$ scalar equations that we can represent in matrix form as 
	\[ \left[\begin{matrix} \Lambda \end{matrix}\middle|\begin{matrix} \lambda
	\end{matrix}\right]\Theta = \bm{0}_{(n+1) \times h}, \] 
	with $\left[\begin{matrix} \Lambda \end{matrix}\middle|\begin{matrix} \lambda
	\end{matrix}\right] \in \mathbb{R}^{(n+1) \times (m+1)(n+1)}$, $\Theta \in \mathbb{R}^{(m+1)(n+1) \times h}$,
	\begin{equation*}
		\left[\begin{matrix} \Lambda \end{matrix}\middle|\begin{matrix} \lambda
		\end{matrix}\right] =  \left[\begin{smallmatrix}
			\nu & \eta & & & & \\
			\nu A & \nu B & \eta & & & \\
			\nu A^2 & \nu AB & \nu B & \eta & & \\
			\vdots & \vdots & & \ddots & \ddots & \\
			\nu A^n & \nu A^{n-1} B & \cdots & & \nu B & \eta 
		\end{smallmatrix}\middle|\begin{smallmatrix} \epsilon \\ \nu c + \epsilon \\ \nu(I+A)c + \epsilon \\ \vdots\\ \nu (\sum_{i=0}^{n-1}A^i) c + \epsilon
		\end{smallmatrix}\right],
	\end{equation*}
	\begin{equation*}
		\Theta =  \begin{bsmallmatrix} X_1 & X_2 & \cdots & X_{h} \\ U_1 & U_2 & \cdots & U_{h} \\
			U_2 & U_3 & \cdots & U_{h+1} \\
			\vdots & \vdots & \ddots & \vdots \\
			U_{n+L} & U_{n+L+1} & \cdots & U_d \\ 1 & 1 & \cdots & 1
		\end{bsmallmatrix} = \begin{bsmallmatrix} \mathcal{H}_1(X_{1:h}) \\ \mathcal{H}_{n+L}(U) \\ \mathds{1}^\intercal
		\end{bsmallmatrix}.
	\end{equation*}
	Now, any linear combination of the rows of $\left[\begin{matrix} \Lambda \end{matrix}\middle|\begin{matrix} \lambda
	\end{matrix}\right]$ also belongs to the left kernel of $\Theta$. In particular, if we select $\beta = \big[ \beta_0 \; \cdots \; \beta_n \big]^{\intercal} \in \mathbb{R}^{n+1}$, $\beta_n = 1$, to contain the coefficients of the characteristic polynomial of $A$, then by Cayley-Hamilton theorem, $\sum_{i=0}^{n} \beta_i A^i = 0$. Then, 
	\begin{equation*}
		\beta^\intercal \left[\begin{matrix} \Lambda \end{matrix}\middle|\begin{matrix} \lambda
		\end{matrix}\right] = \begin{bsmallmatrix} 0 \; & \; \beta_0 \eta+\nu\sum\limits_{i=1}^{n}\beta_iA^{i-1}B \; & \; \cdots & \; \beta_{n-1}\eta+\nu B \; & \; \eta \; & \; \beta^\intercal \lambda
		\end{bsmallmatrix},
	\end{equation*}
	and therefore,
	\begin{equation*}
		\begin{bsmallmatrix}\beta_0 \eta+\nu\sum\limits_{i=1}^{n}\beta_iA^{i-1}B \; & \; \cdots \; & \; \beta_{n-1}\eta+\nu B \; & \; \eta \; & \; \beta^\intercal \lambda
		\end{bsmallmatrix} \begin{bsmallmatrix} \mathcal{H}_{n+L}(U) \\ \mathds{1}^\intercal \end{bsmallmatrix} = 0.
	\end{equation*}
	Note that, since $U$ is persistently exciting of order $n+L+1$, by Proposition \ref{prop:pe} it holds $\mathds{1}^\intercal \notin \textsc{Rowsp}\,\mathcal{H}_{n+L}(U)$ and the dimension of the left kernel of $\begin{bsmallmatrix} \mathcal{H}_{n+L}(U) \\ \mathds{1}^\intercal \end{bsmallmatrix}$ is zero. Therefore, we can conclude that $\eta = 0$. Then, $\nu B = 0$, and we can substitute it in $\nu( \beta_{n-1}B + AB) = 0$ to obtain $\nu AB = 0$. By completing the substitution chain, we get $\nu \begin{bmatrix} B & AB & \cdots & A^{n-1}B
	\end{bmatrix} = 0$ and, by controllability of $(A,B)$, we can conclude that $\nu=0$. Then, by observing that the first row of $\left[\begin{matrix} \Lambda \end{matrix}\middle|\begin{matrix} \lambda
	\end{matrix}\right]$ belongs to the left kernel of $\Theta$, we must conclude that also $\epsilon = 0$. Finally, $\begin{bmatrix} \nu & \eta & \epsilon
\end{bmatrix} = 0$ and (i) is satisfied.
\end{proof}
\begin{cor}
	Under the same assumptions of Theorem~\ref{fundamental lemma for affine sys}, if $L=1$ then $\textsc{Rank}\begin{bsmallmatrix} X^\intercal & U^\intercal & \mathds{1}
	\end{bsmallmatrix}^\intercal = n+m+1$. If, additionally, $d \ge n+2m+1$, then condition \eqref{rank condition} can always be satisfied by an appropriate choice of $W$.
\end{cor}
It follows from Theorem \ref{fundamental lemma for affine sys} that, similarly to the linear setting and even though the augmented pair $(\tilde{A},\tilde{B})$ is uncontrollable, persistency of excitation and controllability of $(A,B)$ are sufficient conditions to obtain a dataset with enough information so that (Theorem \ref{fundamental lemma for affine sys}(i)) the rank condition on the data matrix is satisfied, and (Theorem \ref{fundamental lemma for affine sys}(ii)) $L$-long trajectories are representable as linear combinations of the input-output data. The notable difference is that in case of affine systems we need excitability of one order higher to guarantee that $\mathds{1}^\intercal \alpha = 1$ is satisfied. 

Theorem \ref{fundamental lemma for affine sys}(i) provides an additional insight: the trajectory of a controllable affine system with a persistently exciting input of order $n+2$ can not be contained in any affine subspace of $\mathbb{R}^n$ orthogonal to a unit vector (see Fig. \ref{Fig.affinesubspace}). The same consideration is valid for linear systems by considering Theorem \ref{fundamental lemma for affine sys}(i) with $c=0$.

Lastly, we comment on the use of $W$, which originates from the relaxation of the constraints in \eqref{alternativeLPQfunctions}. Its design is independent from the collected data and, as mentioned in \cite{GoranADP} in deterministic settings, it can be used offline to generate new constraints associated with the same $(x,u)$ pairs but different $w$. For the first time, in the present paper, we provide a design condition on $W$ via \eqref{rank condition} and establish that it must be selected to be independent from the observed state-input trajectories.

\begin{figure}
	\vspace{0.2cm}
	\includegraphics[width=0.95\columnwidth]{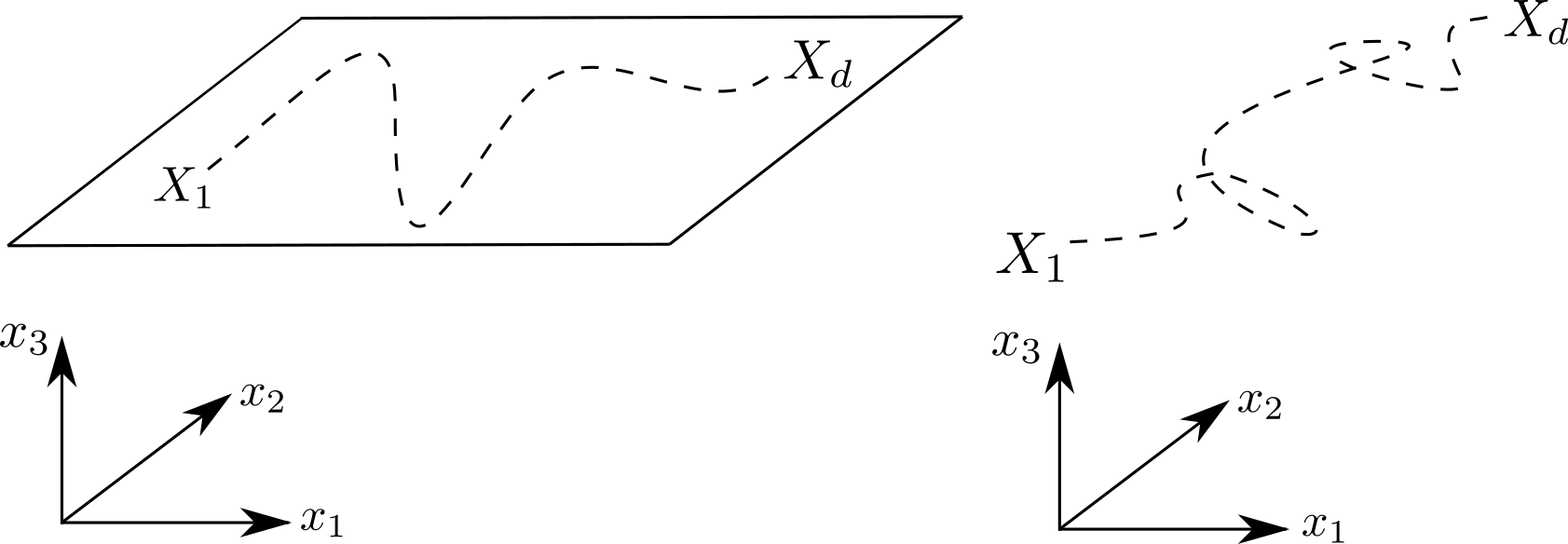}
	\caption{Two example trajectories in $\mathbb{R}^3$ for a controllable system. The one on the left is contained in an affine subspace orthogonal to a unit vector: the associated input sequence is not persistently exciting of order $n+2 = 5$ or higher.}
	\label{Fig.affinesubspace}
\end{figure}

\section{Conclusion}

The present letter focuses on optimal control for affine systems via data-driven linear programming. After introducing the fixed point characterization of three fundamental operators, we show how to synthesize the Bellman inequalities in the LP formulations from data and provide estimators for the associated expected values that preserve the optimal policy. To provide sufficient conditions for the mentioned results, we complete the proof of Willems' fundamental lemma for affine systems. Future research directions will include relaxation of the sufficient conditions in the spirit of \cite{YuExtensionWillems} and online experiment design \cite{HankOnlineExpDesign}.

\bibliographystyle{plain}
\bibliography{Bibliography}

\end{document}